\documentclass[journal]{IEEEtran}
\IEEEoverridecommandlockouts{}
\usepackage{cite}
\usepackage{amsmath,amssymb,amsfonts}
\usepackage{amsthm}
\usepackage{booktabs}
\usepackage{dsfont}
\usepackage{mathrsfs}
\usepackage{graphicx}
\usepackage{pgfplots}
\usepackage{textcomp}
\usepackage{hyperref}
\usepackage{xcolor}
\usepackage[short, nocomma, c3]{optidef} 
\usepackage[capitalize]{cleveref}
\usepackage{outlines}
\usepackage{listings}

\usetikzlibrary{decorations.pathreplacing, calligraphy, arrows, shapes, positioning}
\pgfplotsset{compat=1.16}

\newtheorem{lemma}{Lemma}

\def\BibTeX{{\rm B\kern-.05em{\sc i\kern-.025em b}\kern-.08em
    T\kern-.1667em\lower.7ex\hbox{E}\kern-.125emX}}

\renewcommand{\vec}[1]{\mathbf{#1}}
\newcommand{\transpose}[1]{#1^{\mathsf{T}}}

\newcommand{\cc}[2]{{#1}^{cc}_{#2}}
\newcommand{\cv}[2]{{#1}^{cv}_{#2}}
\newcommand*{\pcc}{\cc{p}{l}}
\newcommand{\tcc}{\cc{t}{j,l}}
\newcommand{\Acc}{\cc{a}{l}}
\newcommand{\Bcc}{\cc{b}{l}}
\newcommand{\Acv}{\cv{a}{l}}
\newcommand{\Bcv}{\cv{b}{l}}
\newcommand{\barAcc}{\cc{\bar{a}}{l}}
\newcommand{\barBcc}{\cc{\bar{b}}{l}}
\newcommand{\barAcv}{\cv{\bar{a}}{l}}
\newcommand{\barBcv}{\cv{\bar{b}}{l}}

\newcommand{\expa}{\alpha_{j,l}}
\newcommand{\xmark}{-}
\newcommand{\cmark}{\checkmark}

\newcommand{\REVONE}[1]{#1}

\begin{document}

\title{Scheduling Battery-Electric Bus Charging under Stochasticity using a Receding-Horizon Approach 
    \thanks{
        This work was supported in part by the National Science Foundation through the Advancing Sustainability through Powered Infrastructure for Roadway Electrification (ASPIRE) Engineering Research Center under Grant EEC-1941524, in part by the Department of Energy through a Prime Award with ABB Ltd., under Grant DE-EE0009194, and in part by PacifiCorp under Contract 3590.

        The authors are with the Department of Electrical and Computer Engineering, Utah State University (USU), Logan, UT 84322, USA\@ (email: \{justin.whitaker derek.redmond greg.droge jake.gunther\}@usu.edu)
    }
}

\author{Justin~Whitaker,~\IEEEmembership{Graduate~Student~Member,~IEEE,} \and Derek~Redmond,~\IEEEmembership{Graduate~Student~Member,~IEEE,} \and Greg~Droge~\IEEEmembership{Member,~IEEE,} \and Jacob~Gunther~\IEEEmembership{Senior~Member,~IEEE}
}

\maketitle

\begin{abstract}
    A significant challenge of adopting battery electric buses into fleets lies in scheduling the charging, which in turn is complicated by considerations such as timing constraints imposed by routes, long charging times, limited numbers of chargers, and utility cost structures.
    This work builds on previous network-flow-based charge scheduling approaches and includes both consumption and demand time-of-use costs while accounting for uncontrolled loads on the same meter.
    Additionally, a variable-rate, non-linear partial charging model compatible with the mixed-integer linear program (MILP) is developed for increased charging fidelity.
    To respond to feedback in an uncertain environment, the resulting MILP is adapted to a hierarchical receding horizon planner that utilizes a static plan for the day as a reference to follow while reacting to stochasticity on a regular basis.
    This receding horizon planner is analyzed with Monte-Carlo techniques alongside two other possible planning methods.
    It is found to provide up to 52\% cost savings compared to a non-time-of-use aware method and significant robustness benefits compared to an optimal open-loop method.
\end{abstract}


\section{Introduction}
\IEEEPARstart{W}{ith} the growing interest in electric vehicles and the increasing push for greener transportation, many organizations and companies are looking to electrify significant portions of their fleet vehicles~\cite{Khan2022}.
In particular, transportation agencies, like the Utah Transit Authority (UTA), are actively exploring and implementing the use of electric buses as replacements for traditional internal combustion engine (ICE) buses.
A battery electric bus (BEB), however, has several difficulties compared to an ICE bus.
The energy storage capacity of a BEB is usually less than an ICE bus, with significantly longer refueling time.
Utility cost structures add complication with multiple components to the cost beyond the cost of energy usage or \textit{consumption}.
Charging several BEBs at once incurs \textit{demand} costs, based on the maximum power draw, which can add significant expense~\cite{Qin2016}.
Furthermore, it is common to have \textit{time-of-use} (TOU), or time-dependent, costs.
Accordingly, scheduling the charging of BEBs can make a significant impact in the day-to-day costs of using BEBs in place of ICE buses.

The stochasticity introduced into the system from real-world operations (traffic delays, discharge variations due to temperature, etc.) can also have significant adverse effect on the cost-aware scheduling of BEB charging.
The complex cost structure, lower energy capacity, and slow refueling times means that BEBs often need to operate near the edge of their capabilities, which in turn amplifies the effects of the stochasticity.
For example, a BEB that can normally operate with charging only at low-cost times, but that is experiencing higher than expected discharging due to temperature, traffic, etc.\ may require charging at inopportune times, such as during high cost TOU periods.
If not properly addressed, this can even result in BEBs being at risk of violating minimum charge level constraints and being unable to continue operation.
This work seeks to provide a method to mitigate the effects of stochasticity in scheduling the charging of BEB fleets while considering the complex utility pricing schedule for the fleet.

Various aspects of the problem have been addressed in the literature; of particular relevance to the proposed work are \textit{utility costs}, \textit{partial charging}, and \textit{stochastice awareness}.
A number of representative works are summarized with respect to these categories in \cref{tab:literature}.
As seen in \cref{tab:literature}, nearly all works consider energy consumption costs.
Fewer consider more complicated additions to the cost structure with~\cite{Zhou2020-Mixed,Rinaldi2020} considering TOU consumption costs while~\cite{He2020,Bagherinezhad2020,Mortensen2023} additionally consider TOU demand costs.
Despite rare consideration, these additional costs can account for a significant percentage of the overall charging costs~\cite{Qin2016,Mortensen2023}

Additionally, several methods of modelling charging behavior have been explored, with the most common methods being to either assume full-charge, utilize a linear partial charge model, or to form a non-linear partial charging model.
As \cref{tab:literature} shows, many works assume a BEB always charges to full capacity when scheduled to charge, e.g.,~\cite{Zhou2020-Mixed,Rinaldi2020,Duan2021,Tang2019}.
This simplifies the problem, but also reduces fidelity and flexibility.
In contrast, a partial charging model was key to~\cite{He2020,Bagherinezhad2020,Mortensen2023} considering demand costs and gaining significant cost improvements.
Works with partial charging use either linear~\cite{Bagherinezhad2020,Chen2016,Mortensen2023} or non-linear~\cite{Zhang2021,He2020,Whitaker23} charging profile models.
The approaches that utilize a non-linear charging profile achieve a higher fidelity model of charging behavior, typically at the cost of increased computation requirements (e.g.,~\cite{Zhang2021}).
One exception is~\cite{He2020}, where a linear program is formed with piecewise-linear charging profiles.
This formulation, however, assumes that a charger is always available when needed, which is often impractical.
Another exception is our prior work in~\cite{Whitaker23}, where a discrete linear time-invariant dynamic system model is used, resulting in an exponential decay non-linear charge model.
This exponential decay model, however, is only a rough approximation of a typical non-linear charging profile.
Non-linear partial charging models that are also computationally attractive are not currently present in the literature without the aforementioned drawbacks.

\begin{table}
    \centering
    \caption{An overview of the features present in the literature.
        \\
        \normalfont{An asterisk (*) indicates a limited formulation.}}\label{tab:literature}
    \begin{tabular}{@{}lllllllll@{}}
        \toprule
                                                      & \multicolumn{3}{c}{\begin{tabular}[c]{@{}c@{}}
                Utility \\ Costs\end{tabular}} &
        \multicolumn{3}{c}{\begin{tabular}[c]{@{}c@{}}Partial \\ Charging\end{tabular}} &
        \multicolumn{2}{c}{\begin{tabular}[c]{@{}c@{}}Stochastic \\ Awareness\end{tabular}}                                                                                                                                  \\ \cmidrule(r){2-4}\cmidrule(rl){5-7}\cmidrule(l){8-9}
                                                      & \rotatebox{90}{Consumption}                   &
        \rotatebox{90}{Demand}                        &
        \rotatebox{90}{TOU}                           &
        \rotatebox{90}{Linear}                        &
        \rotatebox{90}{Non-linear}                    &
        \rotatebox{90}{Variable-Rate}                 &
        \rotatebox{90}{Static}                        &
        \rotatebox{90}{Dynamic}                                                                                                                                                        \\ \midrule
        \cite{Qin2016}                                & \cmark{}                                      & \cmark{}  & \xmark{} & \xmark{}  & \cmark{}* & \xmark{}  & \xmark{} & \xmark{} \\ 
        \cite{Zhou2020-Mixed}                         & \cmark{}                                      & \xmark{}  & \cmark{} & \xmark{}  & \xmark{}  & \xmark{}  & \xmark{} & \xmark{} \\ 
        \cite{Rinaldi2020}                            & \cmark{}                                      & \xmark{}  & \cmark{} & \xmark{}  & \xmark{}  & \xmark{}  & \xmark{} & \xmark{} \\ 
        \cite{He2020}                                 & \cmark{}                                      & \cmark{}  & \cmark{} & \cmark{}* & \xmark{}  & \cmark{}* & \xmark{} & \xmark{} \\ 
        \cite{Bagherinezhad2020}                      & \cmark{}                                      & \cmark{}  & \cmark{} & \cmark{}  & \xmark{}  & \xmark{}  & \xmark{} & \xmark{} \\ 
        \cite{Mortensen2023}                          & \cmark{}                                      & \cmark{}  & \cmark{} & \cmark{}  & \xmark{}  & \cmark{}  & \xmark{} & \xmark{} \\ 
        \cite{Duan2021}                               & \cmark{}                                      & \xmark{}  & \xmark{} & \xmark{}  & \xmark{}  & \xmark{}  & \cmark{} & \xmark{} \\ 
        \cite{Tang2019}                               & \cmark{}                                      & \xmark{}  & \xmark{} & \xmark{}  & \xmark{}  & \xmark{}  & \xmark{} & \cmark{} \\ 
        \cite{Chen2016}                               & \cmark{}                                      & \xmark{}  & \cmark{} & \cmark{}  & \xmark{}  & \xmark{}  & \xmark{} & \xmark{} \\ 
        \cite{Zhang2021}                              & \cmark{}                                      & \xmark{}  & \xmark{} & \xmark{}  & \cmark{}  & \xmark{}  & \xmark{} & \xmark{} \\ 
        \cite{Whitaker23}                             & \cmark{}                                      & \xmark{}  & \cmark{} & \xmark{}  & \cmark{}  & \xmark{}  & \xmark{} & \xmark{} \\ 
        \cite{Zhou2020}                               & \cmark{}                                      & \xmark{}  & \cmark{} & \cmark{}  & \xmark{}  & \xmark{}  & \xmark{} & \xmark{} \\ 
        \cite{Frendo2021}                             & \cmark{}*                                     & \cmark{}* & \xmark{} & \xmark{}  & \cmark{}* & \xmark{}  & \xmark{} & \xmark{} \\ 
        \cite{Jahic2019}                              & \cmark{}*                                     & \cmark{}  & \xmark{} & \xmark{}  & \cmark{}* & \xmark{}  & \xmark{} & \xmark{} \\ 
        \cite{Bie2021}                                & \cmark{}                                      & \xmark{}  & \xmark{} & \xmark{}  & \xmark{}  & \xmark{}  & \cmark{} & \xmark{} \\ 
        \cite{Wang2019}                               & \cmark{}                                      & \xmark{}  & \cmark{} & \xmark{}  & \xmark{}  & \xmark{}  & \xmark{} & \cmark{} \\ 
        \cite{Wang2021}                               & \cmark{}                                      & \xmark{}  & \cmark{} & \xmark{}  & \xmark{}  & \xmark{}  & \xmark{} & \cmark{} \\ 
        Proposed                                      & \cmark{}                                      & \cmark{}  & \cmark{} & \cmark{}  & \cmark{}  & \cmark{}  & \xmark{} & \cmark{}
    \end{tabular}
\end{table}

Most of the works in the state-of-the-art assume the charging rate to be fixed for a charger, as \cref{tab:literature} shows.
As the charging rate directly determines power draw, it can significantly affect demand cost, making it a useful method of decreasing costs.
For example,~\cite{Mortensen2023} uses a set of discrete charge rates with a linear model to reduce the maximum power draw, and consequently the demand cost, but is limited to a few rates.
However, only~\cite{Mortensen2023,He2020} could be found that use a variable rate charging model for BEB charge scheduling, the one limited to a few discrete rates, and the other assuming unlimited chargers.

As \cref{tab:literature} displays, the majority of the BEB charge scheduling methods neglect the presence of real-world noise and variability.
Some works, however, account for this stochasticity in both \textit{static}~\cite{Duan2021,Bie2021} and \textit{dynamic}~\cite{Tang2019,Wang2019,Wang2021} frameworks.
In~\cite{Duan2021,Bie2021,Tang2019}, a static, \textit{a priori} charge scheduling problem is solved to be robust to the expected sources of noise.
The approach taken by~\cite{Tang2019} uses a buffer time at the end of bus routes to maintain feasibility, while~\cite{Duan2021,Bie2021} directly incorporate the characteristics of the noise of energy consumption and route delays in a robust optimization problem.
In contrast,~\cite{Tang2019,Wang2019,Wang2021} utilize dynamic methods that receive feedback from and react to the real-world state.
In~\cite{Tang2019}, the buffered, \textit{a priori} scheduling problem is repeatedly solved with feedback from the environment, while~\cite{Wang2019,Wang2021} use a Markov Decision Process to repeatedly solve for an optimal policy for each bus based on feedback.
However, each of the methods that address stochasticity only consider TOU consumption costs and neglect TOU demand costs.
To the best of our knowledge, no work that considers stochasticity also considers the full TOU consumption and demand cost structure despite the significant effects of these cost components.
Furthermore, each of these works assume a BEB charges to full whenever scheduled with no partial charging, and no work could be found that included a partial charging model.
Accordingly, a significant gap in the literature exists when accounting for the intrinsic stochasticity of the real-world.

In light of this state of the art, this work seeks to address gaps in the literature by providing two primary contributions.
\begin{enumerate}
    \item A novel discrete, time-invariant, piecewise linear formulation is developed to model a non-linear constant-current/constant-voltage charging profile while considering variable-rate charging.
    \item A novel two-stage dynamic receding horizon planning hierarchy is developed to react to feedback from the environment and to fully account for TOU consumption and demand costs.
\end{enumerate}
These contributions are demonstrated within an extensive Monte Carlo simulation environment, demonstrating the ability of the planning hierarchy to cope with the stochastic environment.

The remainder of the paper will proceed with the formulation of the problem as a network flow problem and mixed integer linear program from previous work in \cref{sec:network-flow-MILP}.
\Cref{sec:novel-constraints} continues by detailing the novel modeling of the non-linear, variable-rate charging profile.
This is followed by \cref{sec:costs} which gives the cost formulations and the full static optimization model.
Following this, \cref{sec:receding-horizon} outlines the hierarchical dynamic planning structure, after which \cref{sec:results} presents the experimental findings and their analysis.
\Cref{sec:conclusion} concludes with a summary of the findings and their impact.

\section{A Network Flow Representation of the BEB Charging Schedule}\label{sec:network-flow-MILP}
The assignment of chargers to buses can be abstracted into a network-flow problem.
This network flow problem can be converted into a mixed-integer linear program (MILP) that allows modeling additional aspects and constraints of the system, such as battery charge level dynamics and constraints, power usage and demand costs, and logistic constraints.
This work follows a similar process for creating the network-flow and MILP models as~\cite{Whitaker23} and then extends the model to account for TOU demand and a non-linear CC/CV charging profile.
For the sake of completeness, the full models used in this work are given, with the modifications and extensions to~\cite{Whitaker23} appropriately derived and noted.
This section begins with a brief introduction to relevant graph and network-flow concepts in \Cref{ssec:graph_basics} and their relation to the BEB charge scheduling problem.
Finally, \Cref{ssec:base-constraints} details the constraints of the MILP relating to and extending the network-flow model.

\subsection{Graph and Network Flow From Previous Work~\texorpdfstring{\cite{Whitaker23}}{}}\label{ssec:graph_basics}
A directed graph \(G\) is defined as a set of vertices, \(V\),  and edges, \(E \subseteq V \times V\), i.e. \(G =\{V,E\} \).
Given two vertices in the graph, \(v_1, v_2 \in V\), an edge from \(v_1\) to \(v_2\) implies that the ordered pair \((v_1, v_2)\) is in the edge set, i.e., \((v_1, v_2) \in E\).

With the concept of \textit{flow} along edges, a network-flow problem can be represented by a graph.
For scheduling BEB charging when the BEB route schedules are known \textit{a priori}, each vertex can be considered a state for a charger, with edges representing the action space of chargers.
The amount of flow along an edge indicates the number of chargers performing the corresponding action.
A \textit{source} vertex introduces flow, or chargers, into the network, while \textit{sink} vertices remove flow.
All other vertices are \textit{intermediary} vertices, and maintain a balance of flow in and out of the vertex.
The edges in the graph correspond to one of 1) a charger transitioning to or from charging a BEB, 2) charging a BEB\@, or 3) doing nothing
.
Chargers that are in the same location and have the same characteristics are considered together as a \textit{charger type} in a sub-graph separate from other charger types.
A vertex is only present in a sub-graph when a bus will be available to charge with the corresponding charger type.

A simplified example of such a graph is shown in \cref{fig:bus_nodes} for a single charger type; a full graph contains other similar sub-graphs for each charger type.
\begin{figure}
    \centering
    \includegraphics[width=.9\linewidth]{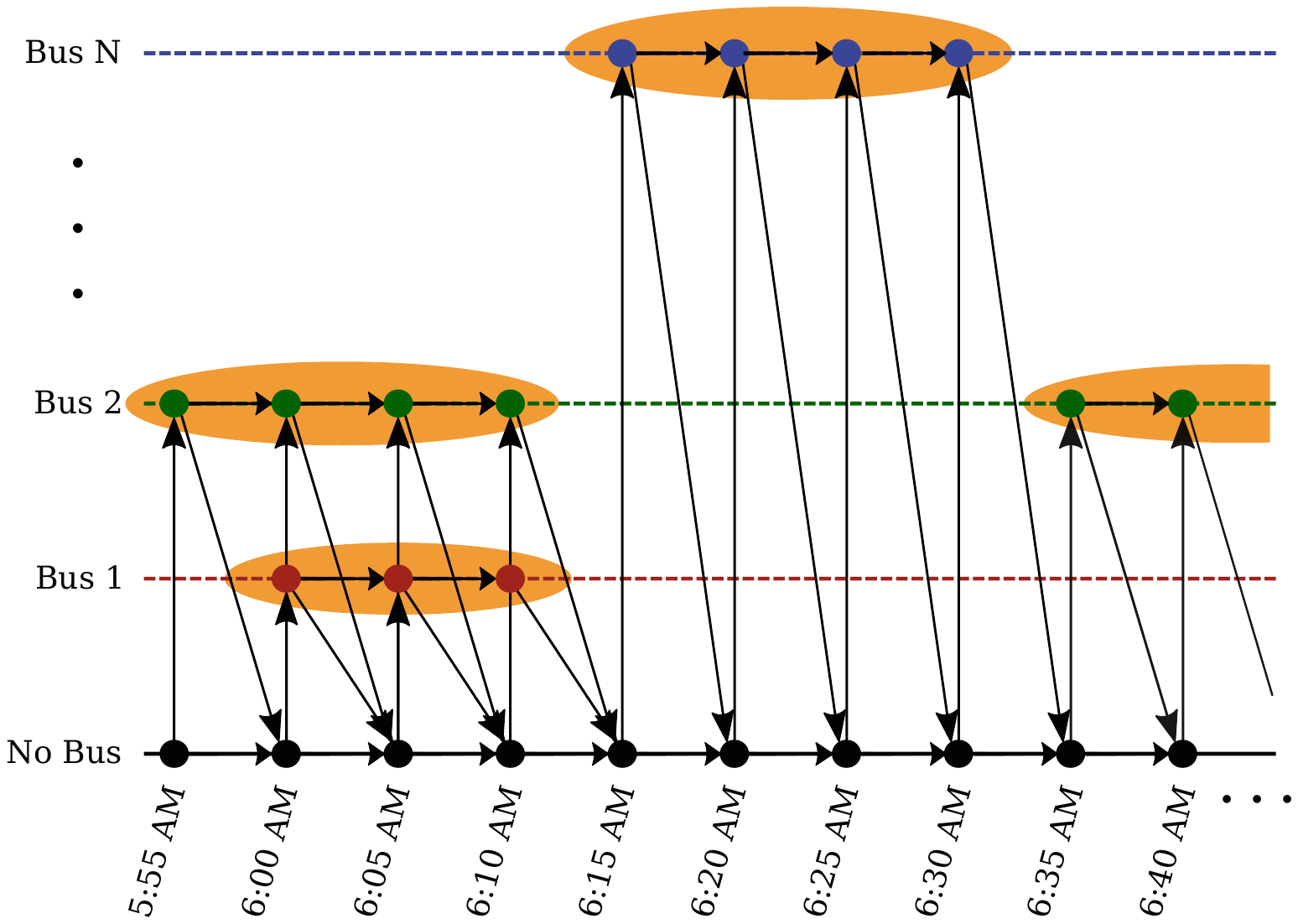}
    \caption{\label{fig:bus_nodes}
        A graph depicting the possible state and transitions for a single charger type over time for a simple {bus charging} example.
        The bottom row corresponds to the charger at rest while the other rows each correspond to charging a particular bus.
        Yellow ovals {show the vertex groups that} are placed around charging windows.
    }
\end{figure}
Vertices that are grouped together (represented with an orange oval) correspond to one \textit{visit} to the station.
This grouping applies across charger types as long as they correspond to the same visit to the station.


\subsection{Constraints from Previous Work~\texorpdfstring{\cite{Whitaker23}}{}}\label{ssec:base-constraints}
In converting to a MILP formulation, constraints are needed to correctly model the problem.
The first set of constraints maintains the relationships defined by the network flow graph.
A second group of constraints is added that ensures that each bus can only be charged by one charger and each charger can only charge one bus.
The final set of constraints is used to track the bus charge levels through charging and discharging.

\subsubsection{Flow balance}
Typical flow balance constraints ensure that the chosen path moves continuously from one vertex to another through the graph.
The vector \(\vec x = \transpose{\begin{bmatrix} x_1 & x_2 & \cdots & x_{n_e} \end{bmatrix}}\) consists of the flow variables for each edge.
Similarly, a vector \(f_l = \transpose{\begin{bmatrix} n_{c_l} & 0 & \cdots & 0 & -n_{c_l}\end{bmatrix}} \) contains the net flow into or out of each vertex for each charger type \(l\in\mathcal{L}\).
The incidence matrix, \(D_l\), (as defined in~\cite{Mesbahi2010}) is created for each sub-graph corresponding to a different charger type.
Each row of the expression \(D_l \vec x\) gives the net flow through a vertex and can be constrained as
\begin{align}\label{eq:flow_constraint}
    D_l \vec x & = f_l & \forall l.
\end{align}

\subsubsection{Vertex Grouping Constraint}
Connecting a BEB to a charger typically requires coordination and time.
Accordingly, it is undesirable to allow a bus to connect to a charger more than one time in a visit to a station.
By grouping the vertices by each bus's visit to the station and only allowing one flow to enter the group this can be enforced.
For the \(p^{th}\) visit by a bus to the station, denote the corresponding grouping of vertices as the set \(\mathcal{V}_p \subset V\).
Then, the index set of edges that enter \(\mathcal{V}_p\) is defined as
\begin{equation}
    \begin{split}
        \mathcal{I}_p = \{i | &e_i \in E \mbox{ where } e_i = (v_m, v_n) \mbox{, } \\
        &v_n \in \mathcal{V}_p \mbox{, and } v_m \in {V \setminus \mathcal{V}_p} \}
    \end{split}.
\end{equation}
The constraint that a bus be charged by at most a single charger when at the station can be written as
\begin{align}\label{eq:group_constraint}
    \sum_{i \in \mathcal{I}_p} x_i & \leq 1 & p & = 1, \ldots, n_g,
\end{align}
where \(n_g\) is the number of groups.
Combined with the constraint that \(x_i\) is integer, \REVONE{\eqref{eq:group_constraint}} ensures that
the number of chargers assigned to a bus during a visit is at most one.

\section{Non-Linear, Variable-Rate Charge Profile}\label{sec:novel-constraints}
A non-linear, variable-rate charging profile model allows for high-fidelity predictions of BEB charge levels, power usage, and subsequent utility costs.
In~\cite{Whitaker23}, a discrete, linear time-invariant system was used to generate an exponential decay non-linear charging profile that loosely approximates a constant-current-constant-voltage (CC-CV) charge profile.
This method is extended herein to form a discrete, piecewise linear time-invariant system that closely approximates an idealized CC-CV charging profile instead of an exponential decay model.

The formulation of the CC-CV system starts with a model in continuous-time that is exactly discretized to arrive at a discrete, piecewise linear time-invariant system in \cref{ssec:cc-cv-base-model}.
\Cref{ssec:cc-cv-variable-rate} adjusts the model to allow for variable rate charging with a close approximation.
Finally, the model is presented in the form of MILP constraints in \cref{ssec:cc-cv-constraints}

\subsection{Modelling a CC-CV Charging Profile in Continuous Time}\label{ssec:cc-cv-base-model}
Forming the CC-CV model is based in a continuous-time piecewise-linear time-invariant dynamic system of the CC-CV charging profile.

\begin{lemma}\label{lem:clti}
    The CC-CV charging profile with switching SOC of \(\eta_j\) for a bus, \(j\), with battery capacity \(E_j\) can be modeled in continuous-time by the piecewise-linear time-invariant continuous dynamic system in the charge level, \(s_j\):
    \begin{equation}
        \dot{s}_{j}(t)
        =
        \begin{cases}
            \Acc s_{j}(t) + \Bcc & 0 \le s_{j}(t) < \eta_j E_j \\
            \Acv s_{j}(t) + \Bcv & \eta_j E_j \le s_{j}(t)
        \end{cases}\label{eq:clti_charging}.
    \end{equation}
    where
    \begin{equation}\label{eq:clti_params}
        \begin{aligned}
            \Acc & = 0      & \Bcc & = \pcc                   \\
            \Acv & = -\expa & \Bcv & = \expa \pcc \tcc + \pcc
        \end{aligned},
    \end{equation}
    and \(\expa \) and \(\pcc \) are parameters of the model derived from the charger characteristics.
\end{lemma}
\begin{proof}

    See \cref{pf:clti}

\end{proof}

As a linear time-invariant system,~\cref{eq:clti_charging} can be discretized exactly for a given discretization step size of \(\delta \), as in~\cite{Hespanha2018}.
This is stated more formally in \Cref{lem:dlti}.
\begin{lemma}\label{lem:dlti}
    The CC-CV charging profile can be modeled by the piecewise-linear time-invariant discrete dynamic system:
    \begin{equation}\label{eq:dlti_charging}
        s_{j,k+1} =
        \begin{cases}
            \barAcc s_{j,k} + \barBcc & 0 \le s_{j,k} < \eta_j E_j \\
            \barAcv s_{j,k} + \barBcv & \eta_j E_j \le s_{j,k}
        \end{cases}
    \end{equation}
    where
    \begin{equation}\label{eq:dlti_params}
        \begin{aligned}
            \barAcc & = e^{\Acc\delta} = 1 & \barBcc & = \int_0^\delta e^{\Acc\tau} d\tau \Bcc = \Bcc\delta                     \\
            \barAcv & = e^{\Acv\delta}     & \barBcv & = \int_0^\delta e^{\Acv\tau} d\tau \Bcv = \frac{(\barAcv - 1)\Bcv}{\Acv}
        \end{aligned}.
    \end{equation}
\end{lemma}
\begin{proof}
    See~\cref{pf:dlti}.
\end{proof}

\subsection{Variable-rate Mixed Integer Charging Model}\label{ssec:cc-cv-variable-rate}
The discrete linear system~\eqref{eq:dlti_charging} only holds while the bus is charging, and differs in specific parameter values between charger types.
To unify the treatment of when a bus is charging and when it is not, a slack variable, \(g_{j,k,l}\), is introduced.
This variable represents the energy gained by bus \(j\) at charger type \(l\) from discrete time \(k\) to \(k+1\) and should take on a value of zero if bus \(j\) did not charge at a charger type \(l\) from time \(k\) to \(k+1\).

For notational convenience, a mapping, \(\sigma(j,k,l)\), is defined that gives the edge index corresponding to bus \(j\) charging with charger type \(l\) starting at time \(k\) (i.e., the edge spans from \(k\) to \(k+1\)).
This provides a method to easily select the edge flow that corresponds to a gain variable.

To achieve the desired behavior of the gain variable several constraints are introduced, beginning with those that enforce a value of zero if the bus is not charging,
\begin{equation}\label{eq:gain_outer_bounds}
    \begin{aligned}
        g_{j,k,l} & \ge 0                     \\
        g_{j,k,l} & \le E_j x_{\sigma(j,k,l)}
    \end{aligned}.
\end{equation}
The constraints in~\eqref{eq:gain_outer_bounds} utilize a big-M style formulation to enforce equality to zero when \(x_{\sigma (j,k,l)} = 0\) but allow \(g_{j,k,l}\) to take on other values when \(x_{\sigma (j,k,l)} = 1\).

When \(x_{\sigma (j,k,l)} = 1\) additional constraints must ensure that \(g_{j,k,l}\) takes on the appropriate value.
If \(x_{\sigma(j,k,l)} = 1\), the gain is equivalent to \(s_{j,k+1} - s_{j,k}\).
Given \(s_{j, k+1}\) in~\eqref{eq:dlti_charging}, \(g_{j,k,l}\) can be written in terms of \(s_{j,k}\) as
\begin{align}
    g_{j,k,l} & =
    \begin{cases}
        (\barAcc - 1) s_{j,k} + \barBcc & 0 \le s_{j,k} < \eta_j E_j \\
        (\barAcv - 1) s_{j,k} + \barBcv & \eta_j E_j \le s_{j,k}
    \end{cases}, \nonumber \\
    \intertext{which, since \(\Acc = 1\), simplifies to}
              & =
    \begin{cases}
        \barBcc                         & 0 \le s_{j,k} < \eta_j E_j \\
        (\barAcv - 1) s_{j,k} + \barBcv & \eta_j E_j \le s_{j,k}
    \end{cases}.
    \label{eq:gain_equality}
\end{align}
This gives a piecewise linear equality constraint to enforce the value of \(g_{j,k,l}\).

While it would be possible to formulate this piecewise linear constraint using specially ordered sets of type two\footnote{See~\cite{Chen2011} for an introduction to specially ordered sets and their application to constraint modelling.}, this would introduce many additional binary variables.
Additionally, in many cases it is desirable to allow the charging rate to vary within the capabilities of the charger.
Variable rate charging allows reducing the power draw and energy usage to only what is strictly necessary, potentially resulting in additional cost reductions.

Accordingly, by only constraining \(g_{j,k,l}\) to be bounded above by~\eqref{eq:gain_equality}, instead of maintaining strict equality, the effective charge rate is allowed to vary.
This results in the ideal constraint
\begin{align}\label{eq:ideal_gain_constraint}
    g_{j,k,l} & \le
    \begin{cases}
        \barBcc                         & 0 \le s_{j,k} < \eta_j E_j \\
        (\barAcv - 1) s_{j,k} + \barBcv & \eta_j E_j \le s_{j,k}
    \end{cases}.
\end{align}
Enforcing this constraint exactly, however, would still require the use of a specially ordered set of type two and additional binary variables to model the switching point of the inequality.

This can be overcome by using a conservative concave approximation of~\eqref{eq:ideal_gain_constraint} that relaxes the switching point of the two lines of~\eqref{eq:ideal_gain_constraint}.
This is more formally stated in~\Cref{lem:approx-gain-constraint}.
\begin{lemma}\label{lem:approx-gain-constraint}
    The inequality~\eqref{eq:ideal_gain_constraint} is conservatively approximated by
    \begin{equation}\label{eq:gain_bound_base}
        \begin{aligned}
            g_{j,k,l} & \le \barBcc                         \\
            g_{j,k,l} & \le (\barAcv - 1) s_{j,k} + \barBcv
        \end{aligned}.
    \end{equation}
    Furthermore, for a given desired approximation error, \(\epsilon_d\), there exists a discretization step size, \(\delta \), for which the approximation error of~\eqref{eq:gain_bound_base}, \(\epsilon \), is bounded by \(\epsilon_d\) (\(i.e., \lvert\epsilon\rvert \le \epsilon_d\)).
\end{lemma}
\begin{proof}
    See~\cref{pf:approx-gain-constraint}
\end{proof}

\subsection{Formulating the Charging Profile Model for Optimization}\label{ssec:cc-cv-constraints}
To fully specify the charge level constraints, several useful sets are defined.
For this purpose, an indicator function, \(\gamma(j,k,l)\), is first introduced
\begin{equation*}
    \gamma (j,k,l) =
    \begin{cases}
        1 & \text{if charging bus \(j\) with charger type \(l\)} \\[-0.4em] &\text{from time \(k\) to \(k+1\) is possible} \\
        0 & \text{otherwise}
    \end{cases}.
\end{equation*}
Two parameterized sets are defined: \(\mathcal{L}(j,k)\) represents the charger types available for a bus \(j\) at a given time \(k\), and \(\mathcal{K}(j,l)\) represents the times when bus \(j\) is able to charge with charger type \(l\).
These are defined as follows
\begin{equation*}
    \begin{aligned}
        \mathcal{L}(j,k) & = \left\{l \mathbin{|} \gamma(j,k,l) = 1 \right\} \\ 
        \mathcal{K}(j,l) & = \left\{k \mathbin{|} \gamma(j,k,l) = 1 \right\} 
    \end{aligned}.
\end{equation*}
Additionally, for notational convenience, the quantity \(g(j,k)\) represents all of the charge gained by bus \(j\) at time \(k\):
\begin{equation*}
    g(j,k) = \sum_{l \in \mathcal{L}(j,k)} g_{j,k,l}.
\end{equation*}

The parameter \(d_{j,k}\) is also introduced to represent the discharge of bus \(j\) from \(k\) to \(k+1\) while on route.
Then, the charge level constraints can be fully specified
\begin{equation}\label{eq:charge_level_constraint}
    s_{j, k+1} =
    \begin{cases}
        s_{j,k} + g(j,k)  &
        \begin{aligned}
             & \text{if } \exists l \text{ s.t.
            }                                   \\[-0.4em]
             & \gamma(j,k,l) = 1
        \end{aligned}           \\
        s_{j,k} - d_{j,k} & \text{otherwise}
    \end{cases} \quad \forall j, \forall k
\end{equation}
with \(g_{j,k,l}\) defined through a combination of \cref{eq:gain_outer_bounds} and \cref{eq:gain_bound_base}
\begin{equation}\label{eq:gain_constraints}
    \left.\begin{aligned}
        g_{j,k,l} & \le \Bcc                            \\
        g_{j,k,l} & \le (\barAcv - 1) s_{j,k} + \barBcv \\
        g_{j,k,l} & \le 0 { + E_j{x_{\sigma(j,k,l)}}}   \\
        g_{j,k,l} & \ge 0                               \\
    \end{aligned}\right\}  
    \begin{aligned}\mbox{\quad} &\forall j, \forall l, \\ &\forall k\in \mathcal{K}(j,l)\end{aligned}.
\end{equation}

\section{Minimization of Utility Costs}\label{sec:costs}
A significant source of operational costs for electric buses comes from charging the batteries and is governed by the electricity utility costs.
Electricity utility costs often consist of two components: a ``consumption'' portion, based on the energy consumed over the billing period, and a ``demand'' portion based on the peak or maximum power draw that occurs within the billing period, e.g.,~\cite{RMPSch8}.
It is also common for both of these components to have time-dependent rates, typically termed time-of-use (TOU) pricing, where the rate for consumption and/or demand cost changes based on the time of day\footnote{
    For example, the Rocky Mountain Power company's Utah rate schedule eight includes both consumption and demand costs with both on- and off-peak TOU pricing for consumption and demand~\cite{RMPSch8}.
}.
The following formulation for the cost in the optimization problem, therefore, seeks to support both consumption and demand costs with TOU pricing for each.

\subsection{Consumption Cost}\label{ssec:energy_cost}
The energy cost is a charge per kilowatt-hour (kWh) of energy consumed.
There are different rates for high- and low-demand times.
Because the gain variables, \(g_{j,k,l}\), already represent the energy in kilowatt-hours used during each discrete time period, this cost is a linear combination of the \(g_{j,k,l}\).
Accordingly, the total energy cost can be calculated as
\begin{equation*}
    \sum_{k\in\mathcal{K}} c_{c,k} \sum_{j\in\mathcal{J}} g(j,k) \end{equation*}
where each \(c_{c, k}\) is a time-dependent cost on the energy consumed in the corresponding time step and, thus, encapsulates the TOU consumption pricing.

\subsection{Baseline Demand Cost}\label{ssec:facilities_cost}
In its simplest form, the demand cost might be based purely on the maximum power draw.
However, it is common for utility costs to include an averaging component to smooth near-instantaneous behavior.
The model herein is based on a 15-minute moving window demand cost from the Rocky Mountain Power Schedule 8~\cite{RMPSch8}.
In other words, the cost is
\begin{equation*}
    c_{b} \max_t (p_{15}(t))
\end{equation*}
where \(c_b\) is the baseline demand cost multiplier and \(p_{15}(t)\) is the average power over the 15 minutes ending at time \(t\).
For the sake of generalization, an arbitrary-sized time period, \(\Delta \), is considered for averaging instead of a 15-minute period.
The average power over a time period \(\Delta \) ending at time \(t\) can be calculated as
\begin{equation*}
    p_{\Delta}(t) = \frac{1}{\Delta} \int_{t-\Delta}^t p(\tau) d\tau.
\end{equation*}
Defining \(e_{\Delta}(t) \triangleq \int_{t-\Delta}^t p(\tau) d\tau \) as the energy over the time period \(\Delta \) ending at \(t\), the average power is
\begin{equation*}
    p_{\Delta}(t) = \frac{e_{\Delta}(t)}{\Delta}.
\end{equation*}

To relate this to the gain variables of the MILP we recall that the gain variable \(g_{j,k,l}\) is the energy used by bus \(j\), with charger type \(l\) over the time period from \(k\) to \(k+1\).
Thus, the total energy used by charging buses over a discretization step, \(\delta \) (from \(k\) to \(k+1\)), is
\begin{equation}
    \label{eq:e_delta_k_buses}
    e_{\delta,k}^{buses} = \sum_{j\in\mathcal{J}} g(j,k).
\end{equation}
However, there may be additional, uncontrolled loads that also contribute to the average power draw.
A prediction of these uncontrolled loads can be incorporated to discourage charging buses during times of high draw from these uncontrolled loads.
It is assumed that a prediction of the discrete-time energy usage of these uncontrolled loads can be obtained as \(e^{load}_{\delta,k}\).
The total energy used over a discretization step is
\begin{equation}\label{eq:e_delta_k}
    e_{\delta,k} = e^{load}_{\delta,k} + e^{buses}_{\delta,k}.
\end{equation}

Accordingly, a discrete average power can be calculated
\begin{equation}\label{eq:discrete-avg-power}
    p_{\Delta,k} = \frac{\sum_{k'=k-m}^{k-1}e_{\delta,k'}\footnotemark}{\Delta}
\end{equation}
with \(m = \frac{\Delta}{\delta}\).
\footnotetext{
    The discretization step size \(\delta \) may not evenly divide \(\Delta \).
    In this case, \(m = \left \lfloor \frac{\Delta}{\delta} \right \rfloor \) and the numerator would be
    \begin{equation}\label{eq:energy-delta-general}
        \sum_{k'=k-m}^{k-1} e_{\delta,k'} + \frac{\Delta \mod{\delta}}{\delta} e_{\delta, k-m-1}.
    \end{equation}
}
The maximum of~\eqref{eq:discrete-avg-power} over all time instances (over all \(k\)) is used for the baseline demand cost, i.e.,
\begin{equation*}
    c_b \max_k p_{\Delta,k}.
\end{equation*}
This \(\max_k p_{\Delta,k}\) can be calculated in the MILP by introducing a slack variable, \(p_{max}\), with the constraints 
\begin{equation*}\label{eq:max-avg-power}
    p_{max} \ge p_{\Delta,k} \mbox{, } \forall k
    .
\end{equation*}
Together, these constraints ensure that \(p_{max}\) is at least as large as the largest \(p_{\Delta,k}\), and the minimization of the baseline demand cost will drive \(p_{max}\) to equal the largest of the \(p_{\Delta,k}\)
.

This allows the baseline demand cost to be written as
\begin{equation*}
    c_b p_{max}
    .
\end{equation*}

\subsection{TOU Demand Cost}\label{ssec:power_cost}
The TOU demand cost is an additional cost on the maximum average power used during high-demand times of the day.
It is very similar to the baseline demand charge except that the TOU demand cost is calculated only for the \(k\) corresponding to the high-demand times (e.g., 6--9 AM and 6--10 PM).
Defining \(\mathcal{K}_{TOU}\) as the set of \(k\) that are within the high-demand times \(p_{max, TOU}\) is introduced as another slack variable to the MILP, 
\begin{equation*}
    p_{max, TOU} \geq p_{\Delta,k} \text{,}\quad k \in \mathcal{K}_{TOU},
\end{equation*}
that, similar to~\eqref{eq:max-avg-power}, will be forced to equality through the cost minimization.
The TOU demand cost can be expressed as
\begin{equation*}
    c_{TOU} p_{max, TOU}
    .
\end{equation*}

\subsection{Base Model}
For notational convenience, several variables are introduced
\begin{align*}
    x          & = \{x_i \forall i\in\mathcal{I}\}                                         \\
    s          & = \{s_{j,k} \forall j\in\mathcal{J}, k\in\mathcal{K}\}                    \\
    g          & = \{g_{j,k,l} \forall j\in\mathcal{J}, k\in\mathcal{K}, l\in\mathcal{L}\} \\
    e_{\delta} & = \{e_{\delta,k} \forall k\in\mathcal{K}\}                                \\
    e_{\Delta} & = \{e_{\Delta,k} \forall k\in\mathcal{K}\}                                \\
    p_{\Delta} & = \{p_{\Delta,k} \forall k\in\mathcal{K}\}
\end{align*}
Combining the constraints and costs as previously described, the base MILP for a static plan is written as follows:
\begin{mini!}<b>
{\substack{x,s,g,\\e_{\delta},e_{\Delta},p_{\Delta},\\p_{max},\\p_{max,TOU}}}{
    \begin{aligned}
    &\sum_{k\in\mathcal{K}} c_{c,k} \sum_{j\in\mathcal{J}} g (j,k) \\
    &\quad + c_b p_{max} + c_{TOU} p_{max,TOU}
    \end{aligned}
}
{\label{eq:base-milp}}{}
\addConstraint{D_l \vec{x} }{= f_l,}{\forall l\in\mathcal{L}}
\addConstraint{\sum_{i \in \mathcal{I}_p} x_i}{\leq 1,}{\forall p \in\mathcal{P}}
\addConstraint{s_{j,k+1}}{
    =\begin{cases}
        s_{j,k} + g (j,k) & \text{if } \exists l \text{ s.t. } \\[-0.4em]
                          & \gamma(j,k,l) = 1                  \\
        s_{j,k} - d_j     & \text{otherwise}
    \end{cases},
    \;}{\begin{aligned} \forall j &\in\mathcal{J} \\ \forall k &\in\mathcal{K} \end{aligned}}
\addConstraint{
    \begin{aligned}
        g_{j,k,l} \\
        g_{j,k,l} \\
        g_{j,k,l} \\
        g_{j,k,l}
    \end{aligned}
}{
    \begin{alignedat}{1}
        & \leq \Bcc \\ 
        & \leq (\barAcv - 1) s_{j,k} + \barBcv \\ 
        & \leq 0 + E_j x_{\sigma(j,k,l)} \\
        & \geq 0 \\
    \end{alignedat},
}{
    \begin{aligned}\forall j &\in\mathcal{J} \\ \forall l &\in\mathcal{L} \\ \forall k &\in \mathcal{K} (j,l)\end{aligned}
}
\addConstraint{e_{\delta,k}}{= e^{load}_{\delta,k} + \sum_{j\in\mathcal{J}} g (j,k),}{\forall k \in \mathcal{K}}
\addConstraint{p_{\Delta,k}}{= \frac{\sum_{k'=k-m}^{k-1}e_{\delta,k'}}{\Delta},}{\forall k\in\mathcal{K}}
\addConstraint{p_{max}}{\ge p_{\Delta,k},}{\forall k \in\mathcal{K}} 
\addConstraint{p_{max,TOU}}{\ge p_{\Delta,k},}{\forall k\in\mathcal{K}_{TOU}} 
\end{mini!}

\section{Receding Horizon Formulation}\label{sec:receding-horizon}
The base model of~\cref{eq:base-milp} is able to generate optimal plans for a given discretization.
These plans, however, are static, or open-loop, meaning that there is no way to react to noise or model mismatches.
However, as the noise is assumed to be zero-mean, the base-model plan is expected to still be a near-optimal reference to attempt to follow, on average.
This suggests a two-level hierarchical approach, as visualized in \cref{fig:hierarchical-planning}, that is the proposed method of this work.
The top-level of the hierarchy consists of the base-model forming a long-term static plan.
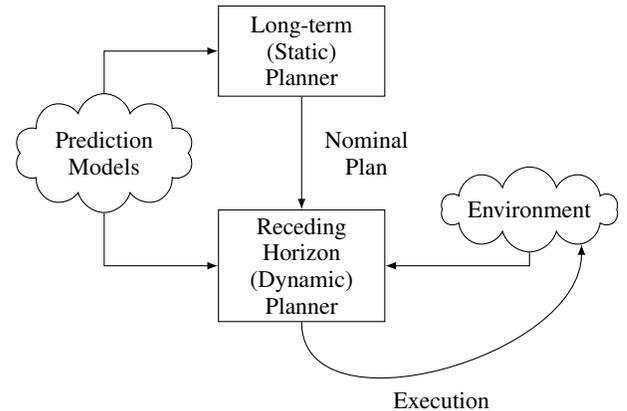
\begin{figure}[b]
    \centering
    \begin{tikzpicture}
        \small
        \tikzstyle{block} = [draw, align=center, minimum height=1.2cm, text width=2cm]
        \node[block](long-term){Long-term (Static) Planner};
        \node[block, below=1.5cm of long-term](receding){Receding Horizon (Dynamic) Planner};
        \draw[-latex] (long-term.south) -- (receding.north) node[midway, right, align=center, text width=1.5cm](between){Nominal Plan};
        \node[cloud, cloud ignores aspect, draw, inner xsep=0pt, right=0.75cm of receding.north east](env){Environment};
        \node[cloud, cloud ignores aspect, align=center, text width=2cm, draw, inner xsep=-6pt, left=1.5cm of between](predict){Prediction Models};

        \draw[-latex] (predict.north) |- (long-term.west);
        \draw[-latex] (predict.south) |- (receding.west);
        \draw[-latex] (env.south) |- (receding.east);
        \draw[-latex] (receding.south) to[in=270, out=270] node[below=0.2cm] {Execution} (env.puff 7);
    \end{tikzpicture}
    \caption{The two-level hierarchical planning method and its interactions with the environment}\label{fig:hierarchical-planning}\label{fig:rapid-planner}
\end{figure}
The top-layer plan uses a coarser resolution to be able to form a plan for the whole day that captures long-term behaviors and constraints.
The bottom-layer is a planner that takes as input the long-term static plan as a reference and feedback from the environment.
The bottom-layer planner regularly outputs a plan with finer resolution than the static planner, allowing rapidly reacting to the environment~\cite{Albus1999}.

In this work a \textit{receding horizon} controller is used as the bottom-level planner.
The receding horizon controller uses the current state of the buses to plan over a relatively short time horizon.
A single time step of this plan is executed, after which feedback from the environment is collected (e.g., current BEB charge levels).
The controller continues by forming a plan over a new horizon, which has ``receded'' one time step beyond the previous horizon.
Receding horizon approaches typically use an optimization model with a ``running'' cost (objective) that scores the state and control trajectory, and a ``terminal'' cost (objective) that scores the state at the end of the horizon.
The terminal cost drives the system toward a set of desired states while the running cost affects the transient behavior.
The dynamics and other state and control constraints are enforced in the optimization model.
It is important to stress that despite using optimization to form a plan for each horizon, there is no general guarantee of optimality with a receding horizon approach.

To obtain the model to use in the receding horizon planner for this problem, the existing optimization model is modified to fit the receding horizon paradigm.
Accordingly, the following adjustments are made to the base model~\eqref{eq:base-milp}
\begin{itemize}
    \item A terminal cost is added to the objective function to drive the controller toward the charge levels given by the top-layer reference plan.
    \item Constraints are added that continue to ensure that a bus charges at most once during a visit to the station.
\end{itemize}

The terminal cost is formed using a penalty on the difference between a bus's terminal charge level in the receding horizon plan and its charge level predicted by the top-level plan.
To add this to the MILP, a slack variable, \(s_{j,err}\), is introduced as
\begin{equation}\label{eq:terminal-cost-error}
    \begin{aligned}
        s_{j,err}
         & \ge \left\{ 
        \begin{aligned}
             & s_{j,T,des} - s_{j,T}     \\
             & s_{j,T} - s_{j,T,des} \ ,
        \end{aligned}
        \right.
    \end{aligned}
\end{equation}
where \(T\) is the terminal time of the horizon, \(s_{j,T}\) is the receding horizon terminal charge level for bus \(j\), and \(s_{j,T,des}\) is the predicted charge level of the bus at \(T\).
Each \(s_{j,err}\) is added to the cost function, and consequentially~\eqref{eq:terminal-cost-error} is equivalent to the 1-norm of the error between \(s_{j,T}\) and \(s_{j,T,des}\).

Ensuring a bus only charges at most once during a visit to the station involves tracking whether a bus has already charged in a visit or not.
By tracking these variables separate from the receding horizon iterations, each receding horizon optimization problem can be constrained to disallow further charging by any bus that has already charged in a station visit.

\section{Results}\label{sec:results}
An empirical analysis of the characteristics of this method is presented to demonstrate the capabilities of this approach.
The general scheduling scenario is presented first, followed by an outline of the noise sources and characteristics.
Finally, Monte-Carlo experiments are performed to characterize the behavior and performance of the receding horizon planning.

\subsection{Bus Charging Scenario}\label{ssec:results_scenario}
Three bus charging scenarios are used to analyze the methodology, each based in the deployment scenarios of the Utah Transit Authority (UTA).
In the near future UTA plans to have 18 BEBs running in the Salt Lake City (SLC), Utah area and 10 BEBs in the Ogden, Utah area running the 2, 209, 220, 509, 602, and 603X UTA routes.
As an example, the 18 SLC schedules are visualized in~\cref{fig:slc_routes}.
\begin{figure}
    \centering
    \includegraphics[width=\linewidth]{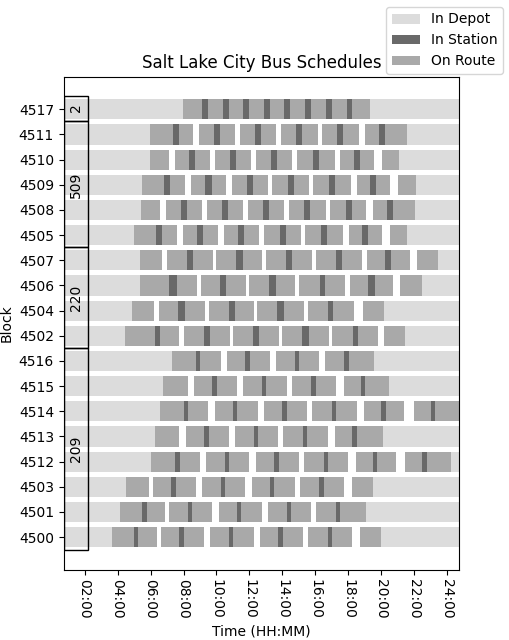}
    \caption{\label{fig:slc_routes}
        The schedules for buses in the Salt Lake City, Utah area that serve routes 2, 209, 220, and 509.
    }
\end{figure}
The experiments in this work utilize the schedules from these two locations as the basis for two of the situations to run in the Monte-Carlo simulations for each location.

The third scenario is one that is randomly generated with 30 buses using a generation method resulting in similar routes to those in the UTA schedules.
This method randomly selects a route length, an amount of time in the station, and a nominal on-route power draw for each bus.
Then, a schedule is formed by alternating on-route and in-station times through the day until the evening (23:00 in this work) when the bus returns to the depot.
The route lengths are selected uniformly from 45 to 150 minutes, the time in station from 20 to 45 minutes, and a nominal power draw over the route from 28 to 36 kW.

The cost structure used in the following experiments is derived from the Rocky Mountain Power (RMP) Schedule 8~\cite{RMPSch8} rate schedule for the winter months.
This rate schedule includes both consumption and demand costs, with a TOU component to each.
The \(c_{c,k}\) is set to \$0.051577 for the time \(k\) in the peak times of 06:00--09:00 and 18:00--22:00 and set to \$0.026216 in the off-peak times.
The set \(\mathcal{K}_{high\textit{-}demand}\) is defined to match these same peak times.
The off-peak demand rate, \(c_b\), is set to \$4.81 with the on-peak rate, \(c_{TOU}\), set to \$13.92
These rates are used for both the reference day plan and the receding horizon planner.

\subsection{Planning parameters}
As described in \cref{sec:receding-horizon}, a static long-term plan is generated by solving the MILP in~\eqref{eq:base-milp}.
For this work a 24-hour period is used in this static plan as it is naturally repeatable.
A discretization step size of 5 minutes is used as it was found to adequately balance computational complexity and solution resolution.
To mitigate the possibility that the upper and lower charge level constraints are violated when noise is introduced in the Monte-Carlo simulations, the static day plan uses a 5\% SOC lower/higher buffer on the upper/lower bound constraints.
To encourage staying near the middle of the battery capacity range, which is better for battery health~\cite{Woody2020}, the initial SOC for all buses is set to 70\%.
The buses are required to end the day at the same 70\% SOC to ensure that the generated schedule is repeatable.
This problem is solved using the off-the-shelf Gurobi solver~\cite{GurobiOptimization2021} limited in solve time to 6 hours.

Similarly, the receding horizon method needs the discretization step size, horizon length and the optimization solve time limit to be determined.
A horizon length of 1 hour with discretization of 3 minutes are used in this work.
These parameters were found to allow sufficient predictive capabilities, while maintaining a relatively short optimization time limit of 10 seconds.
This short time limit is important for performing many Monte Carlo runs within a reasonable time.
In practice, a larger horizon and/or finer discretization could be used.


\subsection{Benefits of variable-rate charging}
\begin{table}
    \centering
    \caption{Variable-rate vs Fixed-rate Cost and Computation}\label{tab:var-rate-comp}
    \begin{tabular}{@{}lllll@{}}
        \toprule
               & \multicolumn{2}{c}{Day Cost} & \multicolumn{2}{c}{\begin{tabular}{@{}c@{}}
                Optimality Gap at 600 s \\ (Solve Time to Optimality)\end{tabular}}                              \\ \cmidrule(r){2-3}\cmidrule(l){4-5}
               & Variable-rate                & Fixed-rate                                    & Variable-rate & Fixed-rate \\
        \midrule
        SLC    & \$648.57                     & \$666.10                                      & (5.76 s)      & 9.92\%     \\
        Ogden  & \$255.46                     & \$292.52                                      & (9.92 s)      & 32.62\%    \\
        Random & \$329.34                     & \$339.52                                      & 3.16\%        & 10.85\%
    \end{tabular}
\end{table}
To validate the benefits of variable-rate charging, the Random, Ogden, and SLC scenarios were compared in both costs and solve complexity under both a fixed-rate and variable-rate formulation.
A linear charging profile was used in both cases, with the variable-rate formulation achieving this by using \(\eta_j\) equal to or greater than the maximum SOC\@.
\Cref{tab:var-rate-comp} show the resulting single-day costs, and the optimality gap (or solve time if solved to optimality) at the end of a 600 second (10 minute) solve time limit used for this comparison.
As seen in \cref{tab:var-rate-comp}, the variable-rate charging results in both lower costs and better computation characteristics compared to the same formulation using fixed-rate charging.

\subsection{Noise types and characteristics}\label{ssec:noise}
To demonstrate the capability of the proposed receding horizon method to cope with noise, a Monte Carlo analysis is performed.
For this, truth models that include noise sources in bus discharging, bus charging, and the station arrival times of each bus are developed.
These truth models and the associated noise characteristics are considered to be not available to the planning methods considered herein.

For the discharging dynamics, a nominal discharge rate is given a random constant bias combined with white noise.
This can be modeled as a discrete linear system for bus \(j\)
\begin{equation}
    \tilde{s}_{j,k+1} = \tilde{s}_{j,k} - d_{j,k} + \beta_{j}^{d} + \nu_{j,k}^d,
\end{equation}
where \(k\) denotes discrete time, \(\nu_{j,k}^d\sim{}\mathscr{N}(0,\sigma_{\nu}^{d})\) is the random variable for the added white noise, and \(\beta_{j}^{d}\sim\mathscr{N}(0,\sigma_{\beta}^{d})\) is the random constant bias.

Similarly, the charging dynamics are given a random constant bias and white noise to result in the following system for bus \(j\)
\begin{equation}
    \tilde{s}_{j,k+1} = \tilde{s}_{j,k} + g(j,k) + \sum_{l\in\mathcal{L}}\left(\beta_{l}^{c} + \nu_{l,k}^{c}\right)x_{\sigma(j,k,l)}
\end{equation}
where \(\nu_{l,k}^{c}\sim{}\mathscr{N}(0,\sigma_{l,\nu}^{c})\) is the added white noise and \(\beta_{l}^{c}\sim\mathscr{N}(0,\sigma_{l,\beta}^{c})\) the random constant bias, varying depending on the charger type \(l\).
The summation over the charger types uses \(x_{\sigma(j,k,l)}\) to select the bias and noise that apply to the charger type that the bus uses in the time step \(k\) to \(k+1\).

The arrival time variations are achieved by perturbing a nominal arrival time, \(t_a\) by white noise such that the new arrival time, \(t_a'\), for a Monte Carlo run is
\begin{equation}
    t_a' = t_a + \nu^{a}
\end{equation}
with \(\nu^{a}\sim\mathscr{N}(0,\sigma^{a})\) being the random variable for the added white noise.
It should be noted that the choice of a normal distribution is used primarily due to its ubiquity, and it is left to future work to explore additional distributions.
\begin{table}
    \centering
    \caption{Noise Parameters}\label{tab:noise-params}
    \begin{tabular}{@{}lllllll@{}}
        \toprule
        \(\sigma_{\nu}^d\) \(\left(\frac{\text{kWh}}{\sqrt{s}}\right)\)
             & \(\sigma_{j, \beta}^d\) (kW)
             & \multicolumn{2}{c}{\(\sigma_{l,\nu}^c\) \(\left(\frac{\text{kWh}}{\sqrt{s}}\right)\)}
             & \multicolumn{2}{c}{\(\sigma_{l,\beta}^c\) (kW)}
             & \(\sigma^a\) (sec)
        \\
             &
             & slow
             & fast
             & slow
             & fast
             &
        \\ \midrule
        0.05 & 1.2                                                                                   & 0.04167 & 0.0833 & 1.2 & 2.4 & 120
    \end{tabular}
\end{table}

The noise parameter values used in this work are given in~\cref{tab:noise-params}.
The charging and discharging values are based in data gathered over the months of June-September on the UTA SLC and Ogden routes previously mentioned.
The arrival time variations are based in data from~\cite{Comi2017}.

\subsection{Monte Carlo experiments}
The noise sources on the discharging, charging, and arrival times are used in a simulation environment to generate Monte Carlo runs.
The Monte Carlo runs are used to analyze the performance of the receding horizon planning method.

To generate one simulation run, each bus is initialized at its nominal initial charge level, without any noise.
The simulation then follows the schedules of the buses with the addition of noise on the arrival time, charge rates, and route discharge.

In each run, the decisions of when to charge, with what charger type, and for how long are made using one of three strategies: the \textit{Qin} strategy is a thresholding strategy based on the work in~\cite{Qin2016} which also approximates a typical bus operator charging strategy; the \textit{Open-Loop} strategy follows the top-level full-day plan as closely as possible\footnote{The Open-Loop strategy honors charge rates and charge stopping times. However, it cannot always start on time due to the bus arriving late.}; and the \textit{Hierarchical} strategy as presented in this work.
For the Qin strategy, the decision to charge or not is made upon arrival at the station, depending on if the SOC is below the threshold (70\% in this work).
A bus stops charging when it reaches the maximum charge level or must leave to continue on its route, whichever comes first.
The Open-Loop method will only charge if the nominal plan indicates to do so, regardless of any stochasticity.
This means the Open-Loop strategy cannot charge more time than the nominal plan indicates, but may charge for less time when buses arrive late.
The Hierarchical strategy, on the other hand, accounts for discrepancies from the nominal plan that it encounters.

A comparison of these charging strategies is made over 50 Monte Carlo runs for each combination of strategy (Qin, Open-Loop, and Hierarchical) and scenario (SLC, Ogden, Random), resulting in a total of 450 runs.

The primary points of comparison across the strategies are the performance (the resulting electric utility costs), and the feasibility and repeatability for each strategy.
Analyzing the feasibility and repeatability balances the desire to reduce costs with the possibility of buses having service interruptions.

As seen in~\cref{tab:mc-costs}, the resulting costs of both optimization-based methods (Open-Loop and Hierarchical) are significantly better than the Qin method, which is similar to charging strategies often used in practice.
\begin{table}
    \centering
    \caption{Strategy and Scenario Daily Cost}\label{tab:mc-costs}
    \begin{tabular}{@{}llll@{}}
        \toprule
               & Qin        & Open-Loop & Hierarchical \\
        \midrule
        SLC    & \$698.10   & \$315.14  & \$384.49     \\   
        Ogden  & \$261.26   & \$116.54  & \$134.87     \\   
        Random & \$1,369.26 & \$552.10  & \$650.96     
    \end{tabular}
\end{table}
While the Open-Loop strategy does provide the best costs, this is primarily due to its lack of feedback and the ``lost'' charging from the instances when a bus arrived late to the station and so could not charge for part of the time for which it was scheduled to charge.

\begin{figure}
    \centering
    \includegraphics[width=\linewidth]{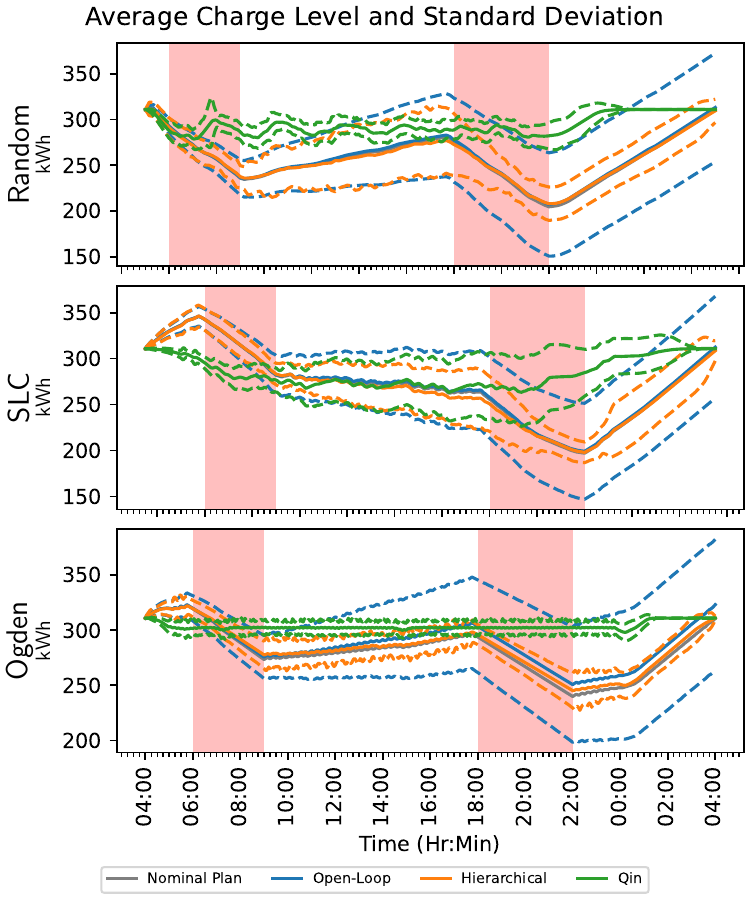}
    \caption{\label{fig:mc-avg-and-std-dev}
        The solid lines are the average charge levels (kWh) across Monte-Carlo runs and buses for each scenario.
        The dotted lines denote the three-\(\sigma \) variation across Monte-Carlo runs and buses (with each bus's charge level zeroed to its mean across Monte-Carlo runs)
        Note that the Qin strategy deviates from the nominal plan significantly during the day (as expected).
        Significantly, the three-\(\sigma \) variations of the Open-Loop strategy continue to grow over time, while the variations of the other two strategies do not.
    }
\end{figure}
Examining \cref{fig:mc-avg-and-std-dev} shows the average charge level over time of each combination of scenario and strategy.
These averages are taken across both the Monte-Carlo runs and the buses simultaneously.
Also shown are the three-\(\sigma \) variations calculated across Monte-Carlo runs and buses zeroed to their individual mean across those Monte-Carlo runs (this zeroing of the mean for each bus removes the effects of inherent differences between routes from the standard deviation, \(\sigma \), calculation).
It is first noted that the Open-Loop and Hierarchical methods both average nearly the same as the average of the nominal plan.
This aligns with expectations as all noise sources are zero-mean.
The Qin strategy, however, deviates significantly from the nominal plan as this strategy makes no reference to the nominal in its execution.

Significantly, \cref{fig:mc-avg-and-std-dev} demonstrates that the three-\(\sigma \) variations of the Open-Loop strategy grow over time in each scenario.
Due to this, the Open-Loop strategy results in much larger terminal variations than the Qin and Hierarchical strategies, especially at the end of the day.
These larger bound values demonstrate a greater sensitivity to the noise for the Open-Loop strategy, which is more likely to result in violations of constraints during operation.
Demonstrating this higher probability of constraint violation, the Open-Loop strategy violated the minimum charge level constraint in 45 of its 150 runs (30\%).
On the other hand, the Hierarchical strategy only dipped below the minimum charge level constraint for one run (\(< 1 \% \)).
Furthermore, that run was able to satisfy all running and terminal constraints with the minimum charge level constraint allowed only an additional 1\%.

The continual spreading of the variations for the Open-Loop strategy is due to the lack of any feedback mechanism to correct the charge levels of each bus back to the nominal values as they drift from nominal over time due to noise.
This underlying drifting of the bus charge levels warns of the possibility of future infeasibility being induced if the Open-Loop strategy is used over multiple days in a row.
To investigate this possibility, a Monte-Carlo analysis was performed over seven days on the Open-Loop and Hierarchical strategies to determine the multi-day feasibility of each strategy.
This multi-day analysis was performed such that the beginning charge levels of each bus for a day were set to that bus's average ending charge from the previous day's Monte-Carlo runs.
A new nominal plan was generated for that day, and then 50 Monte-Carlo runs were performed for each of the two strategies for that day.
On the seventh day of the Random scenario the Open-Loop strategy resulted in a situation where no nominal plan could be created (the problem was infeasible) due to the significant drift of the bus charge levels over time.
It should be noted that creating a new nominal plan did provide some corrective behaviors, but these were not enough to prevent the decay into infeasibility for the Open-Loop strategy.
In contrast, the Hierarchical strategy was capable of continuing to run on the seventh day with no indication that performance or feasibility was degrading.
The significant cost reduction compared to the Qin strategy combined with the much better feasibility characteristics compared to the Open-Loop strategy highlight the strengths of the proposed solution for scheduling the charging of BEBs within uncertain environments.

\section{Conclusion}\label{sec:conclusion}
Charging BEBs cost-effectively is a difficult problem due to the complex cost structure, the limiting battery capacity and charge rate factors, and the need to react to stochasticity in operations.
This work addresses these difficulties with a receding-horizon control strategy that responds to feedback from the environment.
Utilizing an optimization scheme that accounts for both consumption and demand TOU costs, and that references a non-linear, variable-rate partial charging model allows for costs to be optimized with a high degree of fidelity and flexibility.
This results in a significant cost reduction over a method capable of reacting to the environment (up to 52\%), and a more robust solution method compared to one that does not receive feedback.

\section*{Acknowledgements}
Special thanks go to the Utah Transit Authority and Rocky Mountain Power for providing some data used to produce the results of this work.

\appendices%

\section{Proof of~\texorpdfstring{\cref{lem:clti}}{Lemma 1}}\label[appendix]{pf:clti} 
The CC-CV charging scheme consists of charging at a constant current until reaching the switching voltage.
Then, the voltage is held constant while the current decays from the previous current level to zero.
In~\cite{Wu2020}, a model of the CC-CV charging profile was developed as a piecewise function of power draw with constant power \(\pcc \) during the CC phase and exponentially decaying power during the CV phase.

Following~\cite{Wu2020}, the time for a bus \(j\) to reach the switching point between the CC and CV phases from zero charge with charger type \(l\) is \(\tcc \).
This value is based on \(\pcc \), the bus battery's capacity, \(E_j\), and the desired switching point (as a percentage of charge level), \(\eta_j\), according to \(\tcc=\frac{\eta_j E_j}{\pcc} + t_0\).
Assuming, without loss of generality, that the starting time \(t_0 = 0\), from~\cite{Wu2020} the power draw over time, \(p_{j}(t)\), is
\begin{equation}\label{eq:power_draw_base}
	p_{j}(t) =
	\begin{cases}
		\pcc                       & 0 \le t < \tcc \\
		\pcc e^{-\expa (t - \tcc)} & \tcc \le t
	\end{cases}
\end{equation}
with the exponential rate, \(\expa > 0\), being a parameter of the underlying charging system and battery.

Going beyond~\cite{Wu2020} to form a linear dynamic system model begins with determining the bus battery charge level over time, \(s_{j}(t)\), by integrating the power, \(p_{j}(t)\),
\begin{align*}
	s_{j}(t)
	 & =
	\begin{cases}
		\int_{0}^{t}p_{j}(\tau) d\tau                                 & 0 \le t < \tcc \\
		\int_{0}^{\tcc}p_{j}(\tau) d\tau + \int_{\tcc}^{t}p_{j}(\tau) & \tcc \le t
	\end{cases}.
	\\
	\intertext{With the appropriate substitutions from~\eqref{eq:power_draw_base}, this becomes}
	 & =
	\begin{cases}
		\int_{0}^{t}\pcc d\tau                                                        & 0 \le t < \tcc \\
		\int_{0}^{\tcc}\pcc d\tau + \int_{\tcc}^{t}\pcc e^{-\expa (\tau - \tcc)}d\tau & \tcc \le t
	\end{cases}.
\end{align*}
Performing the integration yields a closed-form solution
\begin{align}\label{eq:base-charge-level}
	s_{j}(t) & =
	\begin{cases}
		\pcc t                                                             & 0 \le t < \tcc \\
		\pcc \tcc - \frac{\pcc}{\expa}\left(e^{-\expa (t-\tcc)} - 1\right) & \tcc \le t
	\end{cases}.
\end{align}
\Cref{eq:base-charge-level} can be written as a function of \(p_{j}(t)\).
Focusing on the second case,
\begin{equation}\label{eq:cccv-charge-level-in-terms-of-power}
	s_j (t) = \pcc \tcc - \frac{1}{\expa}\left(\underbrace{\pcc e^{-\expa (t-\tcc)}}_{p_j(t)} - \pcc\right) \quad \tcc \le t.
\end{equation}
Solving~\eqref{eq:cccv-charge-level-in-terms-of-power} for \(p_{j}(t)\)  allows power to be represented as a function of the charge level.
Note that \(\dot{s}_{j}(t)=p_{j}(t)\).
Combining the solution of \(p_{j}(t)\) in terms of \(s_{j}(t)\) with~\eqref{eq:power_draw_base}, a linear dynamic system representation of \(s_{j}(t)\) can be written as
\begin{equation*}
	\dot{s}_{j}(t)
	=
	\begin{cases}
		\pcc                                     & 0 \le t < \tcc \\
		-\expa s_{j}(t) + \expa \pcc \tcc + \pcc & \tcc \le t
	\end{cases}.
\end{equation*}
As \(s_{j}(t)\) is the only time-dependent element, the switching conditions can be written in terms of \(s_{j}(t)\), recalling that \(s_{j}(0) = 0\) and \(s_{j}(\tcc) = \eta_j E_j\),
\begin{align}\label{eq:cccv-clti-pre}
	\dot{s}_{j}(t) & =
	\begin{cases}
		\pcc                                     & 0 \le s_{j}(t) < \eta_j E_j \\
		-\expa s_{j}(t) + \expa \pcc \tcc + \pcc & \eta_j E_j \le s_{j}(t)
	\end{cases}.
\end{align}
The substitutions of~\eqref{eq:clti_params} in~\eqref{eq:cccv-clti-pre} gives the dynamics in~\eqref{eq:clti_charging}.

\section{Proof of~\texorpdfstring{\cref{lem:dlti}}{Lemma 2}}\label[appendix]{pf:dlti} 
A continuous linear time-invariant system, \(\dot{x}(t) = A x(t) + B u(t) \), can be exactly modeled in discrete time (if the \(u(t)\) is constant over a discrete step) as
\begin{equation*}
	x_{k+1} = \bar{A}x_{k} + \bar{B}u_k
\end{equation*}
with
\begin{align*}
	\bar{A} & = e^{A\delta} & \bar{B} & = \int_0^\delta e^{A\tau} d\tau B                    \\
	        &               &         & = A^{-1}(\bar{A} - I)B \quad (A\text{ non-singular})
\end{align*}
as given in~\cite[4.2.1]{Chen1984}.
These relations hold even if \(u(t) = 1\) is constant: i.e., the case in~\eqref{eq:clti_charging}.
Applying these relations to~\eqref{eq:clti_charging} produces the discrete system in~\eqref{eq:dlti_charging}.

\section{Proof of~\texorpdfstring{\cref{lem:approx-gain-constraint}}{Lemma 3}}\label[appendix]{pf:approx-gain-constraint} 
The upper bound~\eqref{eq:ideal_gain_constraint} is composed of two linear functions as visualized in \cref{fig:gain_approximation}.
Equation~\eqref{eq:ideal_gain_constraint} can be approximated by the concave function formed by using the intersection point of the two lines as the switching point.
Concave piecewise-linear functions can be represented exactly by the minimum of all the linear functions (i.e., \(f(x) = \min(f_1(x), f_2(x), \dots, f_n(x))\)).
Accordingly, the concave approximation can be written as
\begin{equation*}
	g_{j,k,l} \le \min\left(\barBcc, (\barAcv - 1) s_{j,k} + \barBcv \right)
\end{equation*}
or, equivalently,
\begin{equation*}
	\begin{aligned}
		g_{j,k,l} & \le \barBcc                          \\
		g_{j,k,l} & \le (\barAcv - 1) s_{j,k} + \barBcv.
	\end{aligned}
\end{equation*}
The quality of this approximation depends on how close the intersection point, \(s_{j,k,l}^{\times}\), is to the true switching point, \(\eta_j E_j\).
\begin{figure}
	\centering
	\begin{tikzpicture}
		\begin{axis}[
				axis lines=left,
				y=0.4cm,
				xlabel=\(s_{j,k}\),
				ylabel=\(g_{j,k,l}\),
				xtick=\empty,
				ytick=\empty,
				legend style={font=\footnotesize, anchor=north east, at={(1,1.01)}},
				font=\footnotesize
			]
			\addplot[domain=0:100, samples=2, color=blue]{4.5};
			\addlegendentry{\(\barBcc\)}
			\addplot[domain=0:100, samples=2, color=red] {-0.1 * x + 10};
			\addlegendentry{\((\barAcv - 1) s_{j,k} + \barBcv\)}
			\addplot+[no markers, dashed] coordinates{(70, 0) (70, 9)};
			\addlegendentry{\(\eta_j E_j\)}
			\node [coordinate, pin=below:{\(s_{j,k}^{\times}\)}]
			at (axis cs:55,4.5) {};
			\draw [
				decorate,
				thick,
				decoration={calligraphic brace, raise=2pt, amplitude=3pt},
			]
			(axis cs:70,4.5) -- (axis cs:70,3)
			node[pos=0.5, right=5pt, black]{\(g_{j,k,l}^e\)};
			\draw [
				decorate,
				thick,
				decoration={calligraphic brace, raise=2pt, amplitude=5pt},
			]
			(axis cs:55,4.5) -- (axis cs:70,4.5)
			node[pos=0.5, above=7pt, black]{\(s_{j,k}^e\)};
		\end{axis}
	\end{tikzpicture}
	\caption{A visualization of the two parts of the approximate gain upper bound constraint (not to scale).
		The blue constant line corresponds to the linear portion of the CC/CV charging profile, while the red corresponds to the exponential portion.
		The intersection point of the two lines is the switching point for the approximation, while the vertical dashed line is the true switching point.
		Thus, the approximation introduces error in between the intersection point and the true switching point, with the most error occurring in the limit approaching the true switching point.
	}\label{fig:gain_approximation}
\end{figure}
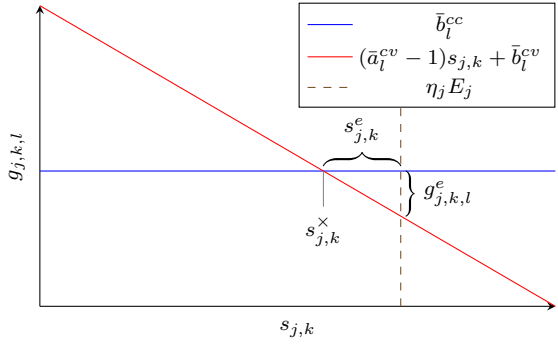

Note that \(\barBcc \) is a constant and \((\barAcv - 1) s_{j,k} + \barBcv \) is linear in \(s_{j,k}\) with \(-1 < \barAcv - 1 < 0\) (i.e., the slope is negative), as \cref{fig:gain_approximation} demonstrates.
Thus, when \(s_{j,k} < s_{j,k}^{\times}\), \(g_{j,k,l} \le \barBcc < (\barAcv - 1) s_{j,k} + \barBcv \) and when \(s_{j,k} > s_{j,k}^{\times}\), \(g_{j,k,l} \le (\barAcv - 1) s_{j,k} + \barBcv < \barBcc \).
Therefore, this approximation matches~\eqref{eq:gain_bound_base} exactly when \(s_{j,k}\) is not in the interval \((s_{j,k}^{\times},\eta_j E_j)\).

The intersection point \(s_{j,k}^{\times}\) occurs at the \(s_{j,k}\) that satisfies
\begin{equation*}
	\barBcc = (\barAcv - 1) s_{j,k} + \barBcv.
\end{equation*}
Solving for \(s_{j,k}\) yields \(s_{j,k}^{\times}\)
\begin{align}
	s_{j,k}^{\times} & = \frac{\barBcc - \barBcv}{\barAcv - 1}, \nonumber                                                                  \\
	\intertext{substituting from~\eqref{eq:dlti_params} gives}
	                 & = \frac{\Bcc\delta - \frac{(e^{\Acv\delta} - 1)\Bcv}{\Acv}}{e^{\Acv\delta} - 1} \nonumber                           \\
	\intertext{and substituting from~\eqref{eq:clti_params} yields}
	                 & = \frac{\pcc\delta - \frac{(e^{-\expa\delta} - 1)(\expa \pcc \tcc + \pcc)}{-\expa}}{e^{-\expa\delta} - 1} \nonumber \\
	                 & = \frac{\pcc\delta}{e^{-\expa\delta} - 1} + \frac{\expa \pcc \tcc + \pcc}{\expa} \nonumber                          \\
	                 & = \frac{\pcc\delta}{e^{-\expa\delta} - 1} +  \pcc \tcc + \frac{\pcc}{\expa}.
	\label{eq:gain_switching_point}
\end{align}
Calculating the error, \(s_{j,k}^e\), between~\eqref{eq:gain_switching_point} and the ideal switching point, \(\eta_j E_j\), is done as
\begin{align}
	s_{j,k}^e & =  \eta_j E_j - s_{j,k}^{\times}                                                           \nonumber        \\
	          & = \eta_j E_j - \frac{\pcc\delta}{e^{-\expa\delta} - 1} - \pcc \tcc - \frac{\pcc}{\expa}.
	\nonumber                                                                                                               \\
	\intertext{Noting that \(\pcc \tcc = \eta_j E_j\) allows cancelling terms}
	          & = -\frac{\pcc\delta}{e^{-\expa\delta} - 1} - \frac{\pcc}{\expa},                          \nonumber         \\
	\intertext{which can be written in terms of \(\barBcc \)}
	          & = -\barBcc\left(\frac{1}{e^{-\expa\delta} - 1}  + \frac{1}{\expa\delta}\right).\label{eq:cccv-charge-error}
\end{align}
Consider the function \(f(z) = \frac{1}{e^{-z} - 1} + \frac{1}{z}\), which is strictly decreasing\footnote{The derivative of \(f(z)\) can be shown to be strictly negative.
} and bounded in range to the open interval \((-1, 0)\)\footnote{The limits as \(z\) approaches positive and negative infinity are -1 and 0, respectively.}.
Taking \(z\) to be \(\expa\delta \) indicates that \(0 < - s_{j,k}^e < \barBcc \), or that \(s_{j,k}^{\times} < \eta_j E_j \), but by no more than \(\barBcc \).

The maximum gain error, \(g_{j,k,l}^e\), between the approximated gain and the ideal value of \(\barBcc \) is related to the error in the switch condition, \(s_{j,k}^e\) by
\begin{align*}
	g_{j,k,l}^e & = (\barAcv - 1) s_{j,k}^e                                                                            \\
	\intertext{as \(\barAcv - 1\) is the slope of the second term of the upper bound (the red line in \cref{fig:gain_approximation}).
		Substituting from~\eqref{eq:clti_params},~\eqref{eq:dlti_params} and~\eqref{eq:cccv-charge-error} gives}
	            & = -(e^{-\expa\delta} - 1)\left(\frac{1}{e^{-\expa\delta} - 1}  + \frac{1}{\expa\delta}\right)\barBcc \\
	            & = -\left(1 - \frac{1-e^{-\expa\delta}}{\expa\delta}\right)\barBcc
\end{align*}
As both \(\expa \) and \(\delta \) are positive, their product must also be positive.
The positivity of the product \(\expa \delta \) ensures that the coefficient of \(\barBcc \) is strictly negative with magnitude strictly less than 1, as \(0 < \frac{1-e^{-z}}{z} < 1\) for \(z > 0\).
Therefore, the maximum gain error, \(g_{j,k,l}^e\), is directly governed by the \(\frac{1-e^{-\expa\delta}}{\expa\delta}\) term.
Since \(\lim_{z\rightarrow 0} \frac{1-e^{-z}}{z} = 1\), an acceptable level of error can be chosen with the selection of \(\delta \) to be sufficiently small.
In other words, for a given desired error, \(\epsilon_d\), a value of \(\delta \) can be chosen sufficiently small according to
\begin{equation*}
	\left(1 - \frac{1 - e^{-\expa\delta}}{\expa\delta}\right) \barBcc \le \epsilon_d
\end{equation*}
ensuring that the actual error \(\epsilon \le \lvert g_{j,k,l}^e\rvert \le \epsilon_d\).
To complete the proof it is noted that when \(s_{j,k}^{\times} < s_{j,k} < \eta_j E_j\) (i.e., when error is introduced in \(g_{j,k,l}\)), the approximate upper bound to \(g_{j,k,l}\) results in a \textit{lower} value than the ideal upper bound.
In other words,~\eqref{eq:gain_bound_base} is a conservative approximation of~\eqref{eq:ideal_gain_constraint}.

\bibliographystyle{./bibliography/IEEEtran}
\bibliography{./bibliography/IEEEabrv,./bibliography/main.bib}


\end{document}